\documentclass[journal,10pt]{IEEEtran}

\usepackage{amssymb}
\usepackage{amsmath,bm}
\usepackage{cite}
\usepackage{url}
\usepackage{xcolor}
\usepackage{cite,graphicx,amsmath,amssymb}
\usepackage{subfigure}

\usepackage{fancyhdr}
\usepackage{mdwmath}
\usepackage{mdwtab}
\usepackage{amsthm}
\usepackage{algorithm}
\usepackage{algpseudocode}
\usepackage{siunitx}

\newtheorem{lem}{Lemma}

\makeatletter
\def\ScaleIfNeeded{%
\ifdim\Gin@nat@width>\linewidth \linewidth \else \Gin@nat@width
\fi } \makeatother

\begin{document}

\title{Delay Minimization in Pinching-Antenna-enabled NOMA-MEC Networks}

%\author{
%\IEEEauthorblockN{ Yuanwei~Liu\IEEEauthorrefmark{1}, Zhijin~Qin\IEEEauthorrefmark{1}, Maged Elkashlan\IEEEauthorrefmark{1}, and  Yue~Gao\IEEEauthorrefmark{1}\\} \IEEEauthorblockA{
%\IEEEauthorrefmark{1} Queen Mary University of London, London, UK\\
%%\IEEEauthorrefmark{2} Lancaster University, Lancaster, UK\\
% } }

\author{Yuan~Ai, Xidong Mu, \emph{Member,~IEEE}, Pengbo~Si, \emph{Senior Member,~IEEE}, Yuanwei~Liu, \emph{Fellow,~IEEE}

\thanks{Y. Ai, P. Si are with the School of Information Science and Technology, Beijing University of Technology, Beijing 100124, China (e-mail: aiyuan@bjut.edu.cn; sipengbo@bjut.edu.cn).	\emph{(Corresponding author: Pengbo Si.)}

X. Mu is with the Centre for Wireless Innovation (CWI), Queen’s University Belfast, Belfast, BT3 9DT, U.K. (e-mail: x.mu@qub.ac.uk).

Y. Liu is with the Department of Electrical and Electronic Engineering, The
University of Hong Kong, Hong Kong (e-mail: yuanwei@hku.hk).
} 
}

\maketitle
\begin{abstract}
This letter proposes a novel pinching antenna systems (PASS) enabled non-orthogonal multiple access (NOMA) multi-access edge computing (MEC) framework. An optimization problem is formulated to minimize the maximum task delay by optimizing offloading ratios, transmit powers, and pinching antenna (PA) positions, subject to constraints on maximum transmit power, user energy budgets, and minimum PA separation to mitigate coupling effects. To address the non-convex problem, a bisection search-based alternating optimization (AO) algorithm is developed, where each subproblem is iteratively solved for a given task delay. Numerical simulations demonstrate that the proposed framework significantly reduces the task delay compared to benchmark schemes.
\end{abstract}

\vspace{-0.5cm} % 减少标题上方的空白
\section{Introduction}
% Introducing the background of Pinching Antenna Systems
Conventional flexible antenna technologies have shown promise in enhancing wireless performance but are hindered by limitations like double-fading effects or restricted operational range\cite{liu2025pinching}. To address these challenges, pinching-antenna systems (PASS) have emerged as a transformative technology, leveraging dielectric waveguides to enable dynamic positioning of pinching antennas (PAs)~\cite{ding2025flexible}. PASS utilize dielectric waveguides to confine electromagnetic waves, significantly reducing propagation losses over extended distances while enabling adaptable line-of-sight connectivity~\cite{ouyang2025array}. {By selectively exciting radiation points along waveguides, PASS enables reconfigurable propagation for programmable wireless environments, particularly through two-state architectures where the system is reconfigured via activation state control~\cite{tyrovolas2025ergodic}.

% Reviewing the research status of integrating MEC with PASS

{Recent research has explored the integration of PASS with multi-access edge computing (MEC) to enhance offloading efficiency. For instance, Liu \emph{et al.}~\cite{liu2025wireless} investigated the optimization of discrete PA activations within PASS, achieving enhanced energy transfer efficiency and improved task offloading rates. Moreover, Tegos \emph{et al.}~\cite{tegos2025minimum} focused on uplink scenarios, proposing a joint optimization of PA positions and resource allocation to maximize data rates. However, these existing studies either assume orthogonal access or focus on system-wide average performance, overlooking the bottleneck users in the network. For delay-sensitive MEC applications, the system performance is strictly limited by the user with the maximum task completion time. In conventional fixed-position antenna systems, these bottleneck users suffer from static deep fading. The potential of PASS to proactively reconfigure the physical propagation environment to alleviate such worst-case delays has not been investigated. Furthermore, in an uplink NOMA scenario, the movement of PAs creates a dynamic interference environment, requiring a joint optimization of physical antenna placement and upper-layer task offloading that goes beyond simple combination strategies.}

To address this research gap, this paper proposes a novel uplink PASS enabled NOMA-MEC framework. {Typical application scenarios include Industrial Internet of Things (IIoT) in smart factories, where PASS combats severe multipath fading to ensure reliable robot control, and multi-user extended reality, where the proposed framework minimizes the worst-case rendering delay to maintain synchronization among users.} {Unlike prior works that focus on system-wide average performance, our primary objective is to minimize the maximum task delay across all users, thereby ensuring strict latency requirements are met for the entire network.} We formulate a joint optimization problem by optimizing offloading ratios, transmit powers, and PA positions, subject to constraints on transmit power, energy budgets, and minimum PA separation. {By leveraging the mobility of PAs along the dielectric waveguide, the proposed framework allows the BS to actively improve the channel conditions of delay-constrained users.} To tackle the non-convexity of the formulated problem, we propose a bisection search-based alternating optimization (AO) algorithm. Extensive simulations demonstrate that the proposed framework effectively reduces the task delay compared to several baselines.
\vspace{-0.4cm} % 减少标题上方的空白
\section{System Model and Problem Formulation}
\subsection{System Model}
\vspace{-0.5cm} % 减少标题上方的空白
\begin{figure} [!htb]
\centering
\includegraphics[width=0.5\columnwidth]{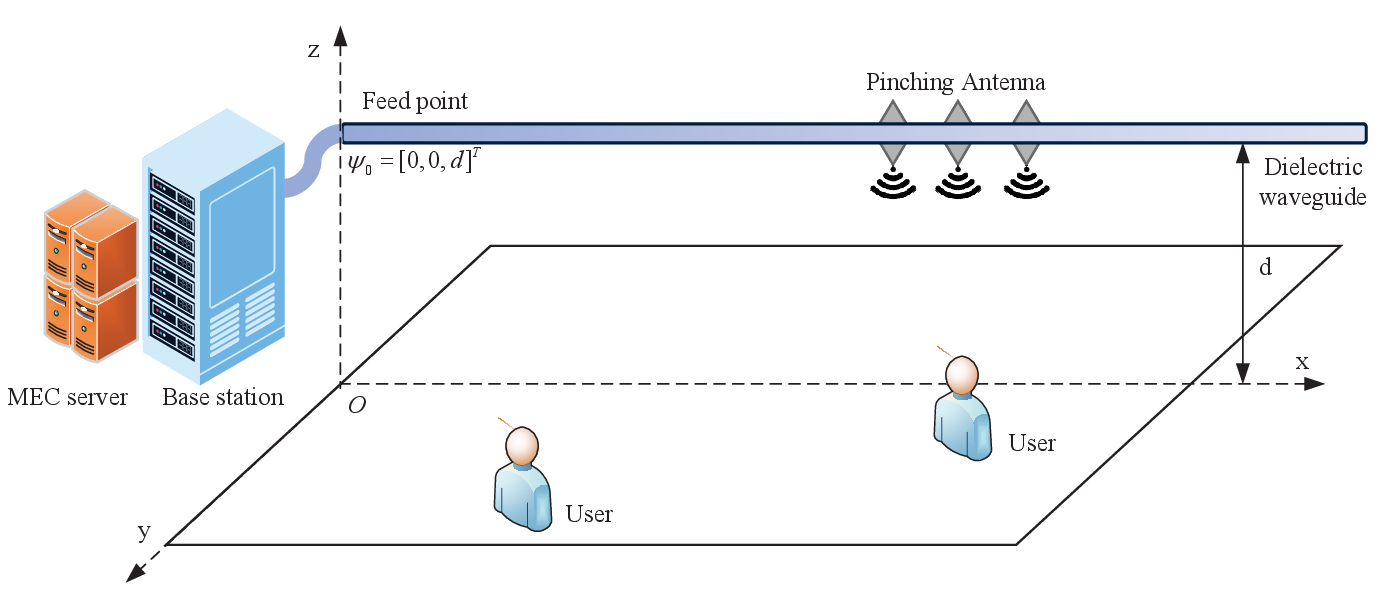}
 \caption{Illustration of a PASS-enabled NOMA-MEC framework with single-waveguide.}
 \label{system_model}
\end{figure}
\vspace{-0.4cm} % 减少标题上方的空白
Consider an uplink NOMA-MEC network comprising \( K \) single-antenna users, indexed by the set \( \mathcal{K} = \{1, 2, \ldots, K\} \), and a BS equipped with a PASS connected to an MEC server, as shown in Fig. \ref{system_model}. The PASS uses a dielectric waveguide of length $L$ meters, having $N$ PAs at positions $\psi_{\mathrm{p}}^n = [x_{\mathrm{p}}^n, 0, d]$, with $\mathbf{x}_{\mathrm{p}} = [x_{\mathrm{p}}^1, \ldots, x_{\mathrm{p}}^N]^T$, constrained by $\mathcal{F} = \{ \mathbf{x}_{\mathrm{p}} \mid 0 \leq x_{\mathrm{p}}^n \leq L, x_{\mathrm{p}}^{n+1} - x_{\mathrm{p}}^n \geq \Delta, \forall n \}$. Users are randomly placed in a square region of side $L$ meters, with positions $\psi_{\mathrm{u}}^k = [x_{\mathrm{u}}^k, y_{\mathrm{u}}^k, 0]$. The channel vector between the $k$-th user and the $N$ PAs is expressed as
\begin{equation}
\mathbf{h}_{k}\left(\mathrm{\mathbf{x}}_{\mathrm{p}}\right)
= {\left[\frac{\eta e^{-j \frac{2 \pi}{\lambda}\left\| \psi_{\mathrm{u}}^k - \psi_{\mathrm{p}}^1\right\|}}{\left\|\psi_{\mathrm{u}}^k - \psi_{\mathrm{p}}^1\right\|}, \ldots, \frac{\eta e^{-j \frac{2 \pi}{\lambda}\left\| \psi_{\mathrm{u}}^k - \psi_{\mathrm{p}}^N\right\|}}{\left\|\psi_{\mathrm{u}}^k - \psi_{\mathrm{p}}^N\right\|} \right]^T},
\end{equation}
where $\eta = \frac{c}{4 \pi f_c}$ is a constant, with $c$ representing the speed of light and $f_c$ the carrier frequency, while $\lambda = \frac{c}{f_c}$ is the wavelength of the carrier signal \cite{ouyang2025array}, and distance is $\|\psi_{\mathrm{u}}^k - \psi_{\mathrm{p}}^n\| = \sqrt{(x_{\mathrm{u}}^k - x_{\mathrm{p}}^n)^2 + D_k^2}$, with $D_k = \sqrt{(y_{\mathrm{u}}^k)^2 + d^2}$. The in-waveguide channel vector is
\begin{equation}
\mathbf{g}(\mathbf{x}_{\mathrm{p}}) = \left[ e^{-j \frac{2 \pi}{\lambda_{\mathrm{g}}} \left\| \psi_{\mathrm{p}}^1 - \psi_0 \right\|}, \ldots, e^{-j \frac{2 \pi}{\lambda_{\mathrm{g}}} \left\| \psi_{\mathrm{p}}^N - \psi_0 \right\|} \right]^T,
\end{equation}
where $\lambda_{\mathrm{g}} = \frac{\lambda}{n_{\mathrm{eff}}}$ is the guided wavelength within the waveguide, and $n_{\mathrm{eff}}$ is the effective refractive index\footnote{{In practical implementations, signals propagating through the waveguide experience attenuation. However, as discussed in \cite{ding2025flexible}, the propagation loss in high-purity dielectric waveguides is significantly smaller than the free-space path loss. Given the short length of the waveguide, the in-waveguide loss is negligible. Thus, following the approach in \cite{ding2025flexible}, we omit this loss term, and the derived results serve as an upper bound on the system performance.}}. {It is worth noting that the uplink signal model relies on the electromagnetic reciprocity of the PASS architecture. Since the dielectric waveguide is a passive linear component, the in-waveguide propagation channel is reciprocal to its downlink counterpart \cite{tegos2025minimum}.} The composite received signal at the BS is formulated as
\begin{equation}
y = \sum_{k=1}^{K}  \mathbf{h}_{k}^T(\mathbf{x}_{\mathrm{p}}) \mathbf{g}(\mathbf{x}_{\mathrm{p}})  \sqrt{\frac{P_k}{N}} s_k + n,
\end{equation}
where $P_k$ is the transmit power of user $k$, $s_k$ is the transmitted symbol satisfying $\mathbb{E}[|s_k|^2] = 1$, and $n \sim \mathcal{CN}(0, \sigma^2)$ represents additive white Gaussian noise with variance $\sigma^2$. {The scaling factor $\sqrt{1/N}$ accounts for the power normalization of the passive combining process within the waveguide. Since the signals from the $N$ PAs are aggregated into a single RF chain, this factor effectively normalizes the receive combining vector to unit norm, ensuring energy conservation and a physically consistent signal-to-noise ratio calculation~\cite{tegos2025minimum}.}

Successive interference cancellation (SIC) is employed at the BS. The effective channel power gain for user \( k \) is defined as \( |v_k(\mathbf{x}_{\mathrm{p}})|^2 = |\mathbf{h}_{k}^T(\mathbf{x}_{\mathrm{p}}) \mathbf{g}(\mathbf{x}_{\mathrm{p}})|^2 \). The decoding order is determined by a permutation \( \pi : \mathcal{K} \to \{1, 2, \ldots, K\} \), where user \( k_m \) satisfies \( \pi(k_m) = m \) at decoding step \( m = K, K-1, \ldots, 1 \). The channel power gains are ordered as
\begin{equation}\label{eq:sic_order}
Q(\mathcal{K}) \stackrel{\triangle}{=}|v_{k_K}(\mathbf{x}_{\mathrm{p}})|^2 \geq |v_{k_{K-1}}(\mathbf{x}_{\mathrm{p}})|^2 \geq ... \geq |v_{k_1}(\mathbf{x}_{\mathrm{p}})|^2.
\end{equation}

For a given decoding order \( \pi \), the BS decodes users from step \(K\) to \(1\), decoding user \(k_m\) at step \(m\), treating signals from users \(j\) with \(\pi(j) < \pi(k_m)\) as interference. The signal-to-interference-plus-noise ratio (SINR) for user \(k\) is
\begin{equation}
    \text{SINR}_k = \frac{P_k |v_k(\mathbf{x}_{\mathrm{p}})|^2}{\sum_{j \in \mathcal{K} : \pi(j) < \pi(k)} P_j |v_j(\mathbf{x}_{\mathrm{p}})|^2 + N \sigma^2},
\end{equation}
where effective channel power gain is explicitly given by
\begin{equation}
|v_k(\mathbf{x}_{\mathrm{p}})|^2 = \left| \sum_{n=1}^N \frac{\eta e^{-j \left[ \frac{2 \pi}{\lambda} \left\| \psi_{\mathrm{u}}^k - \psi_{\mathrm{p}}^n \right\| + \frac{2 \pi}{\lambda_g} \left\| \psi_{\mathrm{p}}^n - \psi_0 \right\| \right]}}{\left\| \psi_{\mathrm{u}}^k - \psi_{\mathrm{p}}^n \right\|}  \right|^2.
\end{equation}

Let $B$ denote the system bandwidth in Hz. The achievable data rate for user $k$ is
\begin{equation}\label{R_m_n_off}
\begin{aligned}
R_{k}%=&B\log_2(1+\Gamma_{m}^{off})\\
=B\log_2\left(\frac{\sum_{i \in \mathcal{K} : \pi(i) \leq \pi(k) } P_i |v_i\left(\mathrm{\mathbf{x}}_{\mathrm{p}}\right)|^2 + N \sigma^2}{\sum_{j \in \mathcal{K} : \pi(j) < \pi(k)} P_j |v_j\left(\mathrm{\mathbf{x}}_{\mathrm{p}}\right)|^2 + N \sigma^2}\right).
\end{aligned}
\end{equation}

%The total achievable data rate across all users in the uplink NOMA system is
%\begin{equation}\label{R_sum}
%	R=\sum\limits_{k=1}^KR_k=B\log_2\left( \frac{\sum_{i=1}^{K} P_i |v_i\left(\mathrm{\mathbf{x}}_{\mathrm{p}}\right)|^2 + N \sigma^2}{N \sigma^2}\right).
%\end{equation} 

Users offload a fraction of computational task to the MEC server, balancing latency and energy via uplink NOMA transmission. Download time from the MEC server is assumed negligible. Offloading delay and energy consumption for user $k$ are
\begin{equation}\label{eq:T_k_off}
T_k^{\text{off}} = \frac{\beta_k L_k}{R_k}, \quad E_k^{\text{off}} = T_k^{\text{off}} P_k.
\end{equation}
where $\beta_k$ is the fraction of the task offloaded by user $k$, and $L_k$ is the total task size in bits. Local computation of the remaining $(1 - \beta_k) L_k$ bits, with CPU frequency $f_k^{\text{loc}}$, yields
\begin{equation}\label{eq:T_k_loc}
T_k^{\text{loc}} = \frac{(1 - \beta_k) L_k C_k}{f_k^{\text{loc}}}, \quad E_k^{\text{c}} = \kappa_k (1 - \beta_k) L_k C_k (f_k^{\text{loc}})^2,
\end{equation}
where $C_k$ is the number of CPU cycles required per bit for user $k$. {The power consumption of the CPU at user $k$ is modeled as $P_k^{\text{CPU}} = \kappa_k (f_k^{\text{loc}})^3$, where $\kappa_k$ is the effective capacitance coefficient for user $k$\cite{fang2020optimal}}.
\vspace{-0.5cm} % 减少标题上方的空白
\subsection{Problem Formulation}
To optimize the performance of the system, we aim to minimize the maximum task completion time across all users while adhering to practical constraints on power, energy, and PA placement. {Regarding placement constraints, practical PASS may utilize discrete locations. Recent studies on two-state PASS \cite{tyrovolas2025how} demonstrate that finite pinching points can achieve near-continuous performance. Thus, our continuous positioning model serves as a fundamental benchmark for practical PASS-enabled MEC design.} The task completion time for user \( k \) is determined by the maximum of the offloading delay and the local computation delay, expressed as
\begin{equation}\label{eq:T_k}
T_k = \max \left\{ T_k^{\text{off}}, T_k^{\text{loc}} \right\},
\end{equation}
where \( T_k^{\text{off}} \) and \( T_k^{\text{loc}} \) are defined in \eqref{eq:T_k_off} and \eqref{eq:T_k_loc}, respectively. There are \( K! \) possible decoding orders for the \( K \) users, as noted in \cite{ren2025pinching}. To address the optimization problem, we perform an exhaustive search over all possible decoding orders to identify the optimal permutation \( \pi \). For a given decoding order \( \pi \), the problem jointly optimizes the offloading ratios \( \bm{\beta} = [\beta_1, \ldots, \beta_K]^T \), transmit powers \( \mathbf{P} = [P_1, \ldots, P_K]^T \), and PA positions \( \mathbf{x}_{\mathrm{p}} = [x_{\mathrm{p}}^1, \ldots, x_{\mathrm{p}}^N]^T \). The optimization problem is formulated as
\begin{subequations}\label{eq:Prob_T_min}
\begin{align}
\min_{\bm{\beta}, \mathbf{P}, \mathbf{x}_{\mathrm{p}}} \quad & {\max_{k \in \mathcal{K}} T_k}\label{eq:obj} \\
\text{s.t.} \quad & \eqref{eq:sic_order}, \quad \forall k, \label{eq:sic_con} \\
& 0 \leq \beta_k \leq 1, \quad \forall k, \label{eq:beta_con} \\
& 0 \leq P_k \leq P_{\max}, \quad \forall k, \label{eq:p_con} \\
& 0 \leq x_{\mathrm{p}}^n \leq L, \quad \forall n, \label{eq:x_p_con} \\
& x_{\mathrm{p}}^{n+1} - x_{\mathrm{p}}^n \geq \Delta, \quad \forall n, \label{eq:delta_con} \\
& E_k^{\text{c}} + E_k^{\text{off}} \leq E_{\max}, \quad \forall k, \label{eq:E_con}
\end{align}
\end{subequations}
where \( P_{\max} \) and \( E_{\max} \) denote the maximum transmit power and energy budget per user, respectively, and \( \Delta \) ensures the minimum separation between PAs. Constraint \eqref{eq:sic_con} enforces the optimal decoding order, \eqref{eq:beta_con} ensures the offloading task ratio is within a valid range, \eqref{eq:p_con} limits the transmit power, \eqref{eq:x_p_con} and \eqref{eq:delta_con} define the feasible set of PA locations along the waveguide, and \eqref{eq:E_con} guarantees that the combined energy consumption for local computation and offloading does not exceed the energy budget for each user. In the following, we focus on a given user decoding order to develop the solution approach.
\vspace{-0.5cm} % 减少标题上方的空白
\section{Proposed Solution}
{To solve \eqref{eq:Prob_T_min}, we extend the time-equalization principle from conventional NOMA-MEC systems \cite{fang2020optimal} to our PASS-enabled framework. Although \cite{fang2020optimal} showed that optimal resource allocation yields equalized offloading times ($T_k^{\text{off}} = T_{k'}^{\text{off}}, \forall k, k'$) under static channels, achieving this in PASS-enabled framework is significantly more complex. Specifically, adjusting the flexible PA positions $\mathbf{x}_{\mathrm{p}}$ to equalize delays creates a dynamic coupling with the SIC decoding order, where position updates may reshuffle channel rankings and invalidate the current SIC sequence. This equalization property is proven by contradiction: if $T_k^{\text{off}} > T_{k'}^{\text{off}}$, adjusting $\beta_k, P_k,$ or $\mathbf{x}_{\mathrm{p}}$ can reduce the maximum completion time without violating constraints. Building on this principle, we handle the positional non-convexity via a multi-resolution search, while Lemma \ref{lem:equivalent_offload_time} derives the equivalent expression in \eqref{eq:T_k_off_equiv}. For analytical simplicity, let $\pi(k) = k$.}

\begin{lem} \label{lem:equivalent_offload_time}
In a multi-user NOMA-MEC PASS where all $K$ users offload tasks within a common transmission time $T = T_k^{\text{off}} = T_{k'}^{\text{off}}, \forall k \neq k'$, the offloading time for user $k$ can be equivalently expressed as
\begin{equation} \label{eq:T_k_off_equiv}
\tilde{T}_k^{\text{off}} = \frac{\sum_{i=1}^{k} \beta_i L_i}{B \log_2 \left( \frac{\sum_{i=1}^{k} P_i |v_i(\mathbf{x}_{\mathrm{p}})|^2 + N \sigma^2}{N \sigma^2} \right)}, \quad \forall k,
\end{equation}
\end{lem}

\begin{proof}

The offloading time for user $k$ is given by $T_k^{\text{off}} = \frac{\beta_k L_k}{R_k}$, where $\beta_k L_k$ is the offloaded task size, and $R_k$ is the achievable data rate for user $k$, as defined in \eqref{R_m_n_off}:
\begin{equation} \label{eq:R_k}
R_k = B \log_2 \left( \frac{\sum_{i=1}^{k} P_i |v_i(\mathbf{x}_{\mathrm{p}})|^2 + N \sigma^2}{\sum_{j=1}^{k-1} P_j |v_j(\mathbf{x}_{\mathrm{p}})|^2 + N \sigma^2} \right).
\end{equation}
Given the condition $T = T_k^{\text{off}} = T_{k'}^{\text{off}}, \forall k \neq k'$, all users share the same offloading time $T$. Thus, for each user $k$, we have: $T = \frac{\beta_k L_k}{R_k}, \forall k.$ This implies that the transmitted data rates satisfy:
\begin{equation} \label{eq:rate_ratio}
\frac{\beta_1 L_1}{R_1} = \frac{\beta_2 L_2}{R_2} = \cdots = \frac{\beta_K L_K}{R_K} = T.
\end{equation}
Using the property of equal ratios, i.e., $\frac{a_1}{b_1} = \frac{a_2}{b_2} = \cdots = \frac{a_K}{b_K} = \frac{\sum_{i=1}^k a_i}{\sum_{i=1}^k b_i}$, we can rewrite \eqref{eq:rate_ratio} as:
\begin{equation} \label{eq:sum_ratio}
T = \frac{\sum_{i=1}^{k} \beta_i L_i}{\sum_{i=1}^{k} R_i}, \quad \forall k.
\end{equation}

After calculation, the expression yields the sum of the achievable rates for the first $k$ users, given by
\begin{equation}
\sum_{i=1}^{k} R_i = B \log_2 \left( \frac{\sum_{j=1}^{k} P_j |v_j(\mathbf{x}_{\mathrm{p}})|^2 + N \sigma^2}{N \sigma^2} \right).
\end{equation}
Substituting into \eqref{eq:sum_ratio}, we obtain:
\begin{equation}
\tilde{T}_k^{\text{off}} = \frac{\sum_{i=1}^{k} \beta_i L_i}{B \log_2 \left( \frac{\sum_{i=1}^{k} P_i |v_i(\mathbf{x}_{\mathrm{p}})|^2 + N \sigma^2}{N \sigma^2} \right)}, \quad \forall k.
\end{equation}
The proof is completed.
\end{proof}
By Lemma \ref{lem:equivalent_offload_time}, we define an auxiliary variable $D_T$ for the common task delay, reformulating \eqref{eq:Prob_T_min} as:
\begin{subequations}\label{eq:Prob_T_transfer}
\begin{align}
\min_{\bm{\beta}, \mathbf{P}, \mathbf{x}_{\mathrm{p}}, D_T} \quad & D_T \\
\text{s.t.} \quad & \frac{\sum_{i=1}^{k} \beta_i L_i}{B \log_2 \left( \frac{\sum_{i=1}^{k} P_i |v_i(\mathbf{x}_{\mathrm{p}})|^2 + N \sigma^2}{N \sigma^2} \right)} \leq D_T, \quad \forall k, \label{eq:R_alphaT_M} \\
& \frac{(1 - \beta_k) L_k C_k}{f_k^{\text{loc}}} \leq D_T, \quad \forall k, \label{eq:bm_alphaT} \\
& \kappa_k(1- \beta_k) L_kC_k (f_k^{loc})^2+D_T P_k\leq E_{\max}, \ \forall k\label{eq:EE_m}\\
& \eqref{eq:sic_con}-\eqref{eq:delta_con}.\label{eq:eq:delta_con}
\end{align}
\end{subequations}
Here, \eqref{eq:R_alphaT_M} enforces the equalized offloading time constraint, \eqref{eq:bm_alphaT} bounds the local computation time, \eqref{eq:EE_m} limits the total energy consumption for local computation and offloading, and the remaining constraints align with those in \eqref{eq:Prob_T_min}. This transformation converts the original max-min problem into a minimization problem with a single objective $D_T$, facilitating efficient algorithmic solutions.

To address the problem \eqref{eq:Prob_T_transfer}, we propose an efficient algorithm combining an outer-layer binary search with an inner-layer AO. The binary search navigates the feasible range of $D_T$ within $[D_T^{\text{min}}, D_T^{\text{max}}]$, exploiting monotonicity to converge to the minimal $D_T$. For a fixed $D_T$, the problem in \eqref{eq:Prob_T_transfer} becomes a feasibility check for $(\bm{\beta}, \mathbf{P}, \mathbf{x}_{\mathrm{p}})$ satisfying \eqref{eq:R_alphaT_M}--\eqref{eq:eq:delta_con}. The inner layer decomposes the problem into subproblems, iteratively optimized alternately. The optimization objectives of the subproblems differ from the original objective of minimizing $D_T$ in \eqref{eq:Prob_T_transfer}. These objectives are not directly equivalent to minimizing $D_T$ but collectively ensure feasibility within the constraints of \eqref{eq:Prob_T_transfer}. By iteratively refining $\bm{\beta}$, $\mathbf{P}$, and $\mathbf{x}_{\mathrm{p}}$ within the outer binary search framework, the algorithm converges to the minimal $D_T$, aligning with the original optimization goal. The subproblems are defined as follows:
\subsubsection{Optimizing $\bm{\beta}$}
With $\mathbf{P}$ and $\mathbf{x}_{\mathrm{p}}$ fixed, the subproblem optimizes the offloading ratios $\bm{\beta}$ to maximize the sum of offloaded task portions:
\begin{subequations}
\begin{align}
    \max_{\bm{\beta}} \quad & \sum_{k=1}^K \beta_k \label{eq:beta_obj} \\
    \text{s.t.} \quad & \sum_{i=1}^{k} \beta_i L_i \leq D_T B \log_2 \Big( \frac{\sum\limits_{i=1}^{k} P_i |v_i(\mathbf{x}_{\mathrm{p}})|^2 + N \sigma^2}{N \sigma^2} \Big), \forall k,  \label{eq:beta_R} \\
    & \beta_k \geq \max \left\{0,  1 - \frac{D_T f_k^{\text{loc}}}{L_k C_k}, 1 - \frac{E_{\max} - D_T P_k}{\kappa_k L_k C_k (f_k^{\text{loc}})^2} \right\}, \forall k, \label{eq:beta_bounds} \\
    &  \beta_k \leq 1, \quad \forall k. \label{eq:beta_upper}
\end{align}
\end{subequations}
The data rate constraint \eqref{eq:beta_R} is derived from \eqref{eq:R_alphaT_M}, and the bounds \eqref{eq:beta_bounds}--\eqref{eq:beta_upper} stem from \eqref{eq:bm_alphaT}, \eqref{eq:EE_m}, and \eqref{eq:beta_con}. As a linear programming (LP) problem, it is solved efficiently using standard solvers (e.g., interior-point methods). If no feasible $\bm{\beta}$ exists, the current $D_T$ is deemed infeasible.
\subsubsection{Optimizing $\mathbf{P}$}
With $\bm{\beta}$ and $\mathbf{x}_{\mathrm{p}}$ fixed, the subproblem optimizes the transmit powers $\mathbf{P}$ to minimize total power consumption:
\begin{subequations}
\begin{align}
    \min_{\mathbf{P}} \quad & \sum_{k=1}^K P_k \label{eq:P_obj} \\
    \text{s.t.} \quad & \sum_{i=1}^{k} P_i |v_i(\mathbf{x}_{\mathrm{p}})|^2 \geq N \sigma^2 \left( 2^{\frac{\sum_{i=1}^{k} \beta_i L_i}{B D_T}} - 1 \right), \forall k, \label{eq:P_R} \\
    & 0 \leq P_k \leq P_{\max}, \quad \forall k, \label{eq:P_bounds_max} \\
    & P_k \leq \frac{E_{\max} - \kappa_k (1 - \beta_k) L_k C_k (f_k^{\text{loc}})^2}{D_T}, \quad \forall k. \label{eq:P_bounds_en}
\end{align}
\end{subequations}
The data rate constraint \eqref{eq:P_R} is derived from \eqref{eq:R_alphaT_M}, and the power bounds \eqref{eq:P_bounds_max}--\eqref{eq:P_bounds_en} are derived from \eqref{eq:EE_m} and \eqref{eq:p_con}. This LP problem is solvable efficiently, and an infeasible $\mathbf{P}$ indicates that the current $D_T$ is not achievable.
\subsubsection{Optimizing $\mathbf{x}_{\mathrm{p}}$}
With $\bm{\beta}$ and $\mathbf{P}$ fixed, the subproblem optimizes the PA positions $\mathbf{x}_{\mathrm{p}}$ to maximize the weighted sum of effective channel gains:
\begin{subequations}
\begin{align}
    \max_{\mathbf{x}_{\mathrm{p}}} \quad & \sum_{k=1}^K P_k |v_k(\mathbf{x}_{\mathrm{p}})|^2 \label{eq:x_p_obj} \\
    \text{s.t.} \quad & \eqref{eq:sic_con}, \quad \forall k, \label{eq:sic_con_sub} \\
    & \sum_{i=1}^{k} P_i |v_i(\mathbf{x}_{\mathrm{p}})|^2 \geq N \sigma^2 \left( 2^{\frac{\sum_{i=1}^{k} \beta_i L_i}{B D_T}} - 1 \right),  \label{eq:x_p_R} \\
    & 0 \leq x_{\mathrm{p}}^n \leq L, \quad \forall n, \label{eq:x_p_bounds} \\
    & x_{\mathrm{p}}^{n+1} - x_{\mathrm{p}}^n \geq \Delta, \quad \forall n = 1, \ldots, N-1. \label{eq:x_p_delta}
\end{align}
\end{subequations}
Due to the non-convexity of $|v_k(\mathbf{x}_{\mathrm{p}})|^2$, we employ element-wise optimization, optimizing each $x_{\mathrm{p}}^n$ while fixing $x_{\mathrm{p}}^m$ for $m \neq n$. For each $n$, the subproblem is:
\begin{subequations}
\begin{align}
    \max_{x_{\mathrm{p}}^n} \quad & \sum_{k=1}^K P_k \left| \sum_{m=1}^N \frac{\eta e^{-j \left[ \frac{2 \pi}{\lambda} \sqrt{(x_{\mathrm{u}}^k - x_{\mathrm{p}}^m)^2 + D_k^2} + \frac{2 \pi}{\lambda_g} x_{\mathrm{p}}^m \right]}}{\sqrt{(x_{\mathrm{u}}^k - x_{\mathrm{p}}^m)^2 + D_k^2}} \right|^2 \label{eq:x_p_n_obj} \\
    \text{s.t.} \quad & x_{\mathrm{p}}^{n-1} + \Delta \leq x_{\mathrm{p}}^n \leq x_{\mathrm{p}}^{n+1} - \Delta, \quad 0 \leq x_{\mathrm{p}}^n \leq L. \label{eq:x_p_n_bounds}
\end{align}
\end{subequations}
{This is a one-dimensional non-convex optimization problem, solved using a multi-resolution search strategy to balance optimality and efficiency. Specifically, we first perform a coarse search to locate the potential region of the global maximum, followed by a fine-grained refinement within the identified region, effectively avoiding the high computational cost of a full high-precision grid search\cite{wang2025modeling}}. After optimizing all $x_{\mathrm{p}}^n$, the solution is validated against \eqref{eq:sic_con_sub} and \eqref{eq:x_p_R}. If infeasible, $D_T$ is marked as infeasible.
\begin{algorithm}[!htb]
    \small
\caption{Task Completion Time Minimization}\label{alg:optimization}
\begin{algorithmic}[1]
\State Initialize \( D_T^{\text{min}} \), \( D_T^{\text{max}} \), \( \epsilon = 10^{-4} \), \( I_{\text{max}} = 20 \), \( \epsilon_x = 10^{-4} \).
\While{\((D_T^{\text{max}} - D_T^{\text{min}}) / D_T^{\text{max}} > \epsilon\)}
    \State \( D_T \gets (D_T^{\text{min}} + D_T^{\text{max}})/2 \), Initialize: $\bm{\beta}^{(0)}, \mathbf{P}^{(0)}, \mathbf{x}_{\mathrm{p}}^{(0)}$
    \For{\( i = 1 \) to \( I_{\text{max}} \)}
        \State Solve \eqref{eq:beta_obj}--\eqref{eq:beta_upper} for \(\bm{\beta}\) (LP)
        \If{infeasible}
            \State Break
        \EndIf
        \State Solve \eqref{eq:P_obj}--\eqref{eq:P_bounds_en} for $\mathbf{P}$ (LP)
        \If{infeasible}
            \State Break
        \EndIf
        \For{\( n = 1 \) to \( N \)}
            \State Solve \eqref{eq:x_p_n_obj}--\eqref{eq:x_p_n_bounds} for \( x_{\mathrm{p}}^n \) {(multi-resolution search)}
        \EndFor
        \If{ \eqref{eq:sic_con_sub} and \eqref{eq:x_p_R} violated}
            \State Break
        \EndIf
        \If{\(\|\bm{\beta} - \bm{\beta}_{\text{old}}\|_2 / \|\bm{\beta}_{\text{old}}\|_2 < \epsilon\) and \(\|\mathbf{P} - \mathbf{P}_{\text{old}}\|_2 / \|\mathbf{P}_{\text{old}}\|_2 < \epsilon\) and \(\|\mathbf{x}_{\mathrm{p}} - \mathbf{x}_{\mathrm{p},\text{old}}\|_2 / \|\mathbf{x}_{\mathrm{p},\text{old}}\|_2 < \epsilon\)}
            \State Break
        \EndIf
    \EndFor
    \If{all constraints \eqref{eq:R_alphaT_M}--\eqref{eq:eq:delta_con} satisfied}
        \State Store \((\bm{\beta}, \mathbf{P}, \mathbf{x}_{\mathrm{p}})\), update \( D_T^{\text{max}} \gets D_T \)
    \Else
        \State Update \( D_T^{\text{min}} \gets D_T \)
    \EndIf
\EndWhile
\State \Return \( D_T^{\text{min}}, \bm{\beta}, \mathbf{P}, \mathbf{x}_{\mathrm{p}} \)
\end{algorithmic}
\end{algorithm}

The algorithm, outlined in Algorithm \ref{alg:optimization}, ensures convergence via an outer-layer binary search and inner-layer AO algorithm. The algorithm’s complexity comprises an outer binary search with $O(\log((D_T^{\text{max}} - D_T^{\text{min}})/\epsilon))$ iterations and inner AO algorithm. Each iteration optimizes $\bm{\beta}$ and $\mathbf{P}$ via LPs with $O(K^3)$ complexity for $K$ variables and constraints, and $\mathbf{x}_{\mathrm{p}}$ via $N$ one-dimensional searches. {By adopting the multi-resolution search, the complexity for updating PAs is reduced to $O(N (S_c + S_f))$, where $S_c$ and $S_f$ denote the number of points in the coarse and fine search stages, respectively. With $I_{\text{max}}$ iterations per feasibility check, total complexity is $O(\log((D_T^{\text{max}} - D_T^{\text{min}})/\epsilon) \cdot I_{\text{max}} (K^3 + N (S_c + S_f)))$.}

The optimization problem \eqref{eq:Prob_T_min} is solved for a fixed decoding order \(\pi\). To obtain the global optimum, we evaluate all decoding orders \(\pi \in \Pi\), where \(\Pi\) is the set of all permutations of user indices \(\{1, 2, \ldots, K\}\). For each \(\pi\), Algorithm \ref{alg:optimization} computes the minimal task delay \(D_T\) and the corresponding \(\bm{\beta}\), \(\mathbf{P}\), and \(\mathbf{x}_{\mathrm{p}}\). The global optimal solution is determined by selecting the configuration with the smallest \(D_T\)\footnote{{The computational complexity of the proposed solution is primarily governed by the $O(K!)$ exhaustive search over decoding orders. While the execution time grows factorially with $K$, the algorithm remains scalable for practical NOMA-MEC systems where user clusters are typically small.}}. 
\vspace{-0.4cm} % 减少标题上方的空白
\section{Simulation Results}
\vspace{-0.2cm} % 减少标题上方的空白
\label{sec:simulation_results}
% Setting up the simulation environment
Simulations are conducted in a \SI{15}{\meter} square region with $K=2$ randomly placed users. {We focus on the $K=2$ scenario, representing a practical NOMA pair often used to limit SIC complexity, although the proposed algorithm is applicable to the general multi-user case.} The BS uses a dielectric waveguide of length $L$, having $N=4$ PAs with minimum separation $\Delta = \lambda/2$, positioned at height $d = \SI{3}{\meter}$ with feed point at $[0, 0, d]$. Key parameters include: carrier frequency $f_c = \SI{28}{\giga\hertz}$, refractive index $n_{\text{eff}} = 1.4$, bandwidth $B = \SI{1}{\mega\hertz}$, noise power spectral density $\SI{-174}{\deci\bel m/\hertz}$, task size $L_k = \SI{1}{\mega\bit}$, local computation $C_k = 10^3$ cycles/bit, local CPU frequency $f_k^{\text{loc}} = 10^9$ cycles/s, capacitance coefficient $\kappa_k = 10^{-27}$, maximum power $P_{\max} = \SI{10}{\deci\bel m}$, energy budget $E_{\max} = \SI{0.2}{\joule}$, and algorithm settings $\epsilon = \epsilon_x = 10^{-4}$, $I_{\text{max}} = 20$. Performance is compared against: {(1) a conventional MIMO system representing a real antenna array, where an $N$-element uniform linear array is fixed at the location $[0, 0, d]$ with $\lambda/2$ spacing, utilizing a single active RF chain for analog beamforming;} (2) an FDMA-based MEC PASS with equal bandwidth allocation; {(3) a TDMA-based MEC PASS with equal time allocation, where users offload tasks sequentially in dedicated time slots using the same waveguide setup. Note that to ensure a fair comparison regarding total energy consumption, the transmit power in the TDMA scheme is scaled by a factor of $K$ during the active time slots.}

\vspace{-0.3cm} % 减少标题上方的空白
\begin{figure}[!htb]
    \centering
    \includegraphics[width=0.6\columnwidth]{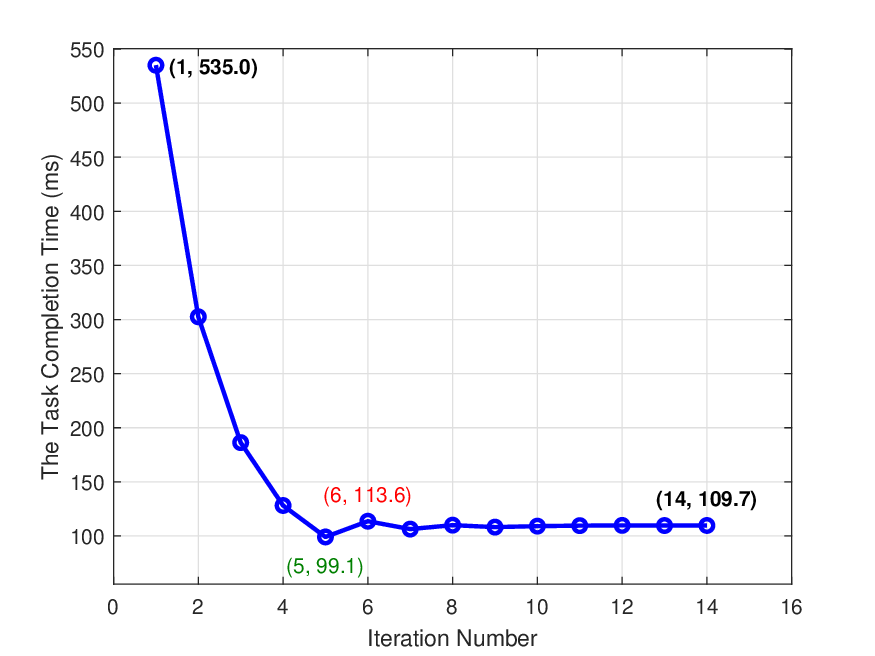}
    \caption{Convergence of the task completion time $D_T$.}
    \label{fig:convergence}
\end{figure}
% Subsection for convergence behavior
%\subsection{Convergence Behavior}
\vspace{-0.3cm} % 减少标题上方的空白
Fig.~\ref{fig:convergence} shows the convergence of Algorithm~\ref{alg:optimization} for $N=4$ PAs, plotting $D_T$ against outer-layer binary search iterations. $D_T$ decreases from \SI{535.0}{\milli\second} to \SI{109.7}{\milli\second} within 14 iterations. A temporary increase in $D_T$ from \SI{99.1}{\milli\second} at iteration 5 to \SI{113.6}{\milli\second} at iteration 6 occurs because the candidate $D_T = \SI{99.1}{\milli\second}$ is infeasible, failing to satisfy constraints in \eqref{eq:Prob_T_transfer}. {The convergence rate is primarily governed by the binary search precision and the system scale. Specifically, the number of outer iterations is strictly determined by the initial interval $[D_T^{\min}, D_T^{\max}]$ and the tolerance $\epsilon$. Regarding the inner loop, increasing $K$ raises the dimensionality of the LPs, while a larger $N$ expands the solution space for position updates, potentially requiring more AO iterations to satisfy the feasibility constraints.}

\vspace{-0.3cm} % 减少标题上方的空白
\begin{figure}[!htb]
    \centering
    \includegraphics[width=0.6\columnwidth]{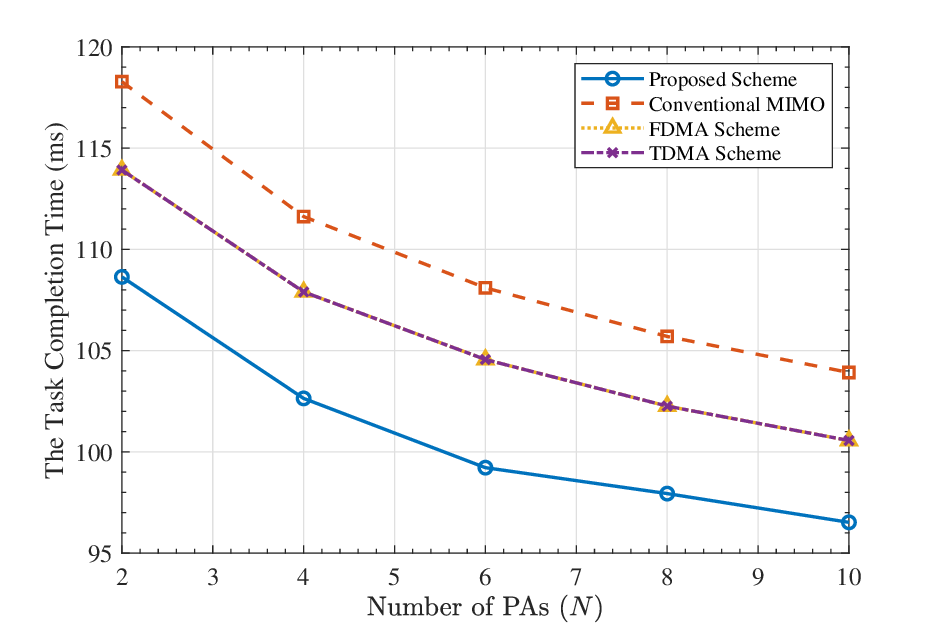}
    \caption{Task completion time $D_T$ versus number of PAs.}
    \label{fig:delay_vs_N}
\end{figure}

\vspace{-0.3cm} % 减少标题上方的空白
{Fig.~\ref{fig:delay_vs_N} compares $D_T$ for the proposed NOMA-MEC PASS, conventional MIMO, and the OMA baselines (TDMA and FDMA) as PAs vary from 2 to 10, averaged over 30 random user placements. The proposed scheme achieves the lowest $D_T$, outperforming MIMO and OMA schemes due to reduced path loss. It is observed that the TDMA and FDMA schemes exhibit similar performance due to the energy-fair setup.} Increasing $N$ further reduces $D_T$ by enhancing PA spatial distribution. {Nevertheless, as proved in \cite{ouyang2025array}, the array gain of PASS is non-monotonic with $N$, where an excessively large number of PAs can lead to performance degradation due to power dilution and mutual coupling effects.}

\vspace{-0.3cm} % 减少标题上方的空白
\begin{figure}[!htb]
    \centering
    \includegraphics[width=0.6\columnwidth]{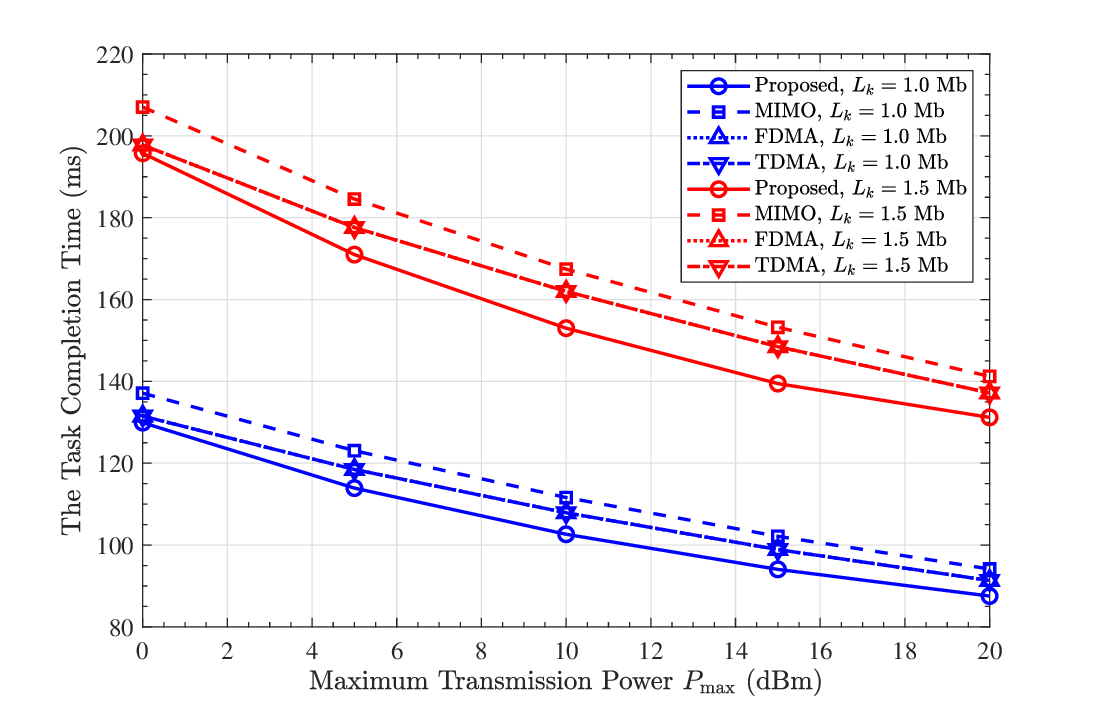}
    \caption{Task completion time $D_T$ versus maximum transmit power $P_{\max}$ for task sizes $L_k = \SI{1}{\mega\bit}$ and $L_k = \SI{1.5}{\mega\bit}$.}
    \label{fig:delay_vs_Pmax}
\end{figure}

\vspace{-0.3cm} % 减少标题上方的空白
Fig.~\ref{fig:delay_vs_Pmax} plots $D_T$ versus maximum transmit power $P_{\max}$ (\SI{0}{\deci\bel m} to \SI{20}{\deci\bel m}) for task sizes $L_k = \SI{1}{\mega\bit}$ and \SI{1.5}{\mega\bit}, averaged over 30 random user placements. {The proposed scheme consistently yields lower $D_T$ than MIMO and the OMA schemes (TDMA/FDMA), with advantages amplified for larger $L_k$. Although the OMA baselines benefit from orthogonal resource allocation, they are limited by the static division of resources.}

\vspace{-0.4cm} % 减少标题上方的空白
\section{Conclusions}
\vspace{-0.2cm} % 减少标题上方的空白
This letter presented a novel uplink PASS enabled NOMA-MEC framework. We formulated a joint optimization problem to minimize the maximum task delay by optimizing offloading ratios, transmit powers, and PA positions, subject to practical constraints. To tackle the problem, we developed a bisection search-based AO algorithm, iteratively solving each subproblem for a given task delay. Extensive simulations demonstrated that the proposed framework significantly outperformed several baselines, achieving substantial reductions in task delay.

\bibliographystyle{IEEEtran}
\bibliography{IEEEabrv,egbib}

@STRING{IEEE_J_COML       = "{IEEE} Commun. Lett."}

@STRING{IEEE_J_COM        = "{IEEE} Trans. Commun."}

@STRING{IEEE_M_COM        = "{IEEE} Commun. Mag."}

@STRING{IEEE_J_WCOML       = "{IEEE} Wireless Commun. Lett."}

@article{fang2020optimal,
  title={Optimal resource allocation for delay minimization in {NOMA-MEC} networks},
  author={Fang, Fang and Xu, Yanqing and Ding, Zhiguo and Shen, Chao and Peng, Mugen and Karagiannidis, George K},
  journal=IEEE_J_COM,
  volume={68},
  number={12},
  pages={7867--7881},
  year={2020},
  publisher={IEEE}
}

@article{ding2025flexible,
  title={Flexible-antenna systems: A pinching-antenna perspective},
  author={Ding, Zhiguo and Schober, Robert and Poor, H Vincent},
  journal=IEEE_J_COM,
  year={2025},
  volume={73},
  number={10},
  pages={9236-9253},
  publisher={IEEE}
}

@article{liu2025pinching,
  title={Pinching-Antenna Systems: Architecture Designs, Opportunities, and Outlook},
  author={Liu, Yuanwei and Wang, Zhaolin and Mu, Xidong and Ouyang, Chongjun and Xu, Xiaoxia and Ding, Zhiguo},
  journal=IEEE_M_COM,
  year={2025},
  publisher={IEEE}
}

@article{tegos2025minimum,
  title={Minimum data rate maximization for uplink pinching-antenna systems},
  author={Tegos, Sotiris A and Diamantoulakis, Panagiotis D and Ding, Zhiguo and Karagiannidis, George K},
  journal=IEEE_J_WCOML,
  year={2025},
  volume={14},
  number={5},
  pages={1516-1520},
  publisher={IEEE}
}

@article{ouyang2025array,
  title={{Array gain for pinching-antenna systems ({PASS})}},
  author={Ouyang, Chongjun and Wang, Zhaolin and Liu, Yuanwei and Ding, Zhiguo},
  journal=IEEE_J_COML,
  year={2025},
  volume={29},
  number={6},
  pages={1471-1475},
  publisher={IEEE}
}

@article{liu2025wireless,
  title={Wireless Powered {MEC} Systems via Discrete Pinching Antennas: {TDMA} versus {NOMA}},
  author={Liu, Peng and Fei, Zesong and Hua, Meng and Chen, Guangji and Wang, Xinyi and Liu, Ruiqi},
  journal={arXiv preprint arXiv:2509.20908},
  year={2025}
}

@article{wang2025modeling,
  title={Modeling and beamforming optimization for pinching-antenna systems},
  author={Wang, Zhaolin and Ouyang, Chongjun and Mu, Xidong and Liu, Yuanwei and Ding, Zhiguo},
  journal=IEEE_J_COM,
  year={2025},
  volume={73},
  number={12},
  pages={13904-13919},
  publisher={IEEE}
}

@article{ren2025pinching,
  title={Pinching-Antenna Systems ({PASS}) Meet Multiple Access: {NOMA} or {OMA}?},
  author={Ren, Qiao and Mu, Xidong and Lin, Siyu and Liu, Yuanwei},
  journal={arXiv preprint arXiv:2506.13490},
  year={2025}
}

@article{tyrovolas2025how,
  title={{How many pinching antennas are enough?}},
  author={Dimitrios Tyrovolas and Sotiris A. Tegos and Yue Xiao and Panagiotis D. Diamantoulakis and Sotiris Ioannidis and Christos K. Liaskos and George K. Karagiannidis and Stylianos D. Asimonis},
  journal={arXiv preprint arXiv:2512.18761},
  year={2025}
}

@article{tyrovolas2025ergodic,
  title={{Ergodic rate analysis of two-state pinching-antenna systems}},
  author={Tyrovolas, Dimitrios and Tegos, Sotiris A and Xiao, Yue and Diamantoulakis, Panagiotis D and Ioannidis, Sotiris and Liaskos, Christos and Karagiannidis, George K and Asimonis, Stylianos D},
  journal={arXiv preprint arXiv:2511.01798},
  year={2025}
}

\end{document}